\documentclass[11pt]{amsart}
\usepackage{amsmath,amsthm,amssymb,graphicx,enumerate,natbib,color}
\newtheorem{theorem}{Theorem}[section]
\newtheorem{lemma}[theorem]{Lemma}

\newtheorem{assumption}[theorem]{Assumption}
\theoremstyle{remark}
\newtheorem{remark}[theorem]{Remark}
\newtheorem{example}[theorem]{Example}
\numberwithin{equation}{section}
\begin{document}

\author[S. Ghosh]{Souvik Ghosh}
\address{Department of Statistics\\
Columbia University \\
New York, NY 10027.}
\email{ghosh@stat.columbia.edu}

\author[S. Ghosh]{Soumyadip Ghosh}
\address{Business Analytics and Math Sciences\\ 
IBM TJ Watson Research Center\\ 
Yorktown Heights, NY 10598.}
\email{ghoshs@us.ibm.com}

\date{\today}
\title[Long Latencies in a Cloud Server]{A Strong Law for the Rate of Growth of Long Latency Periods in Cloud Computing Service}  


\begin{abstract} 
Cloud-computing shares a common pool of resources across customers at
a scale that is orders of magnitude larger than traditional multi-user
systems. Constituent physical compute servers
are allocated multiple ``virtual
machines'' (VM) to serve simultaneously. Each VM user should ideally be
unaffected by others' demand. Naturally, this environment produces new
challenges for the service providers in meeting customer expectations while
extracting an efficient utilization from server resources. We study a new cloud service metric that measures prolonged latency or delay
suffered by customers. We model the workload process of a cloud server
and analyze the process as the customer population grows. The capacity
required to ensure that average workload does not exceed a threshold
over long segments is characterized. This can be used by cloud
operators to provide service guarantees on avoiding long durations of
latency. As part of the analysis, we provide a uniform large-deviation
principle for collections of random variables that is of independent
interest.
 \end{abstract}

\maketitle


{\bf AMS (2000) Subject Classification:} Primary 60F10, Secondary 60F15, 60G99

\begin{section}{Introduction} \label{sec:intro}
Cloud computing is a paradigm shift of multiple orders of magnitude in the pursuit of extracting greater utilization of server resources
while serving the computing needs of a large collection of customers. This has been made possible primarily by the concept of workload {\em virtualization} wherein individual users operate on {\em virtual machines} (VMs), each with modest resource requirements, and multiple VMs 
are served by a single large computing server. Cloud service providers achieve greater utilization by over-provisioning VMs on compute nodes, acting on the assumption that rarely will multiple customers simultaneously require large quantities of resources. 

The resources required over time by a user is a stochastic process, modeled here as a discrete-time moving-average (MA) process. We allow for a heterogeneous population of customers, where they are partitioned only by their statistical/stochastic behaviour but are considered equal in terms of priority of service. Service guarantees currently provided by cloud computing providers (Amazon Web Services' EC2 , Google's  Web Toolkit, Microsoft's Azure etc.) are weak: Service Level Agreements (SLAs) are available only for quick initial provisioning of a new VM from a user onto a compute node, but no guarantees are provided on the quality of service experienced by the customer over time. Large organizations with significant computing requirements, who are willing to pay for good service guarantees, are thus wary of using this architecture for any activity beyond their non-critical desktop usage; see~\cite{li:intelwp:09,mendler:yankee:2010}.
This in particular impedes large-scale adoption of cloud computing for time-critical and resource-intensive workloads. 

New techniques need to be developed to address the challenge of estimating performance from the user's perspective in this computing paradigm. 
A key performance indicator in multi-user systems measures the {\em latency} suffered by users. Latency occurs when access to computing resources is throttled because the total quantity of one or more resource required (CPU cycles, Memory space, IO bandwidth etc.) by all the VMs exceed the server's capacity. Then, under the most commonly used form of processor sharing discipline, all customers on the server are provisioned proportionately lower resources than they had requested and thus are said to experience latency. Suppose the server is allocated a capacity that maintains a steady per-customer average $C_p$ above its expected value. Even if $C_p$ is a large number, there will be time segments during which the average workload of the server will exceed the total capacity.  Applications that are intolerant to latency are discouraged from being put on clouds in the absence of Service Level Agreements that penalize their incidence (\cite{li:intelwp:09}). Therefore, for a company that wishes to guarantee its customers availability of the server's resources, it is important to understand how large and frequent such long time segments of continued latency can be. We provide a framework to construct such estimates. In particular, we use this framework to estimate the time till the first observation of continued latencies of a given large time length, and its dual, the largest period of latency experienced within a given time. Cloud service operators can utilize this technique to create SLA contracts. In addition, the relationship between the expected first observation time and the per-customer average capacity  can help design system improvements to minimize SLA violations. An operator may also provide differentiated service to customers, where those willing to pay for better guarantees can be put on an isolated sub-cloud with capacity provisioning tailored to their growth, usage and the agreed upon SLA contract.

Our framework is built on analyzing {\em long strange segments} (see definitions \eqref{eq:str.seg.gen} and \eqref{eq:lss:tr}) of the underlying workload process of the cloud server; refer \cite{Arratia:1990p6736} and \cite{ghosh:samorodnitsky:2010} for a review. A standard technique for analyzing the rate of growth of long strange segments for stationary processes involves an associated large deviation principle (see discussion at the end of Section~\ref{sec:mainresult}).  While standard probabilistic models (for example, queues) operate on stationary processes, the cloud 
workload process is non-stationary (see definition in Section~\ref{sec:mainresult}). 
This is because the total number of virtual machines in the cloud environment increases over time. This is a consequence of the fact that VMs are software artifacts that are inexpensive to instantiate and operate, and so client organizations tend to encourage large-scale adoption and persistent usage of the VMs within their organization. In addition, a major new technological innovation allows fast migration of VMs between individual physical servers within the same cloud infrastructure. Thus, the cloud service environment is better modeled to consist of larger logical servers that each continually grow in capacity in order to serve a continually growing population of users,  which yields a non-stationary workload process.

The standard large deviation tools that are vital to the analysis of long strange segments of stationary processes are thus not useful for our non-stationary workload process. This process  
however has a certain structure that can be gainfully exploited.
To take advantage of this, we develop a tool for proving uniform large deviation principle that in its most general form applies to collections of random variables that satisfy certain regulatory conditions (see Theorem~\ref{thm:gen:uldp} in Section~\ref{sec:uldp}).  This tool, which is of independent interest, plays a   crucial role in  proving 
Theorem~\ref{thm:lss}, the main result of this paper, which provides a strong law characterization of the rate of growth of duration of latency periods as a function of the Fenchel-Legendre transform of the log moment generation functions of the underlying process. The conditions imposed by the uniform large deviation principle (Theorem~\ref{thm:uldp}) admit many common models for computer workloads. 

To summarize, the main contributions of this paper are:
\begin{enumerate}[a)]
\item We provide a tool for proving uniform large deviation principle for a collection of sequences of probability measures.  Recall that the Gartner-Eliis Theorem is a very helpful device for proving large deviation principle for a single sequence of probability measures; refer \cite{gartner:1977}, \cite{ellis:1984} and \citep[Theorem 2.3.6, p.44]{dembo:zeitouni:1998}.  We view Theorem~\ref{thm:gen:uldp} as an analogue of the Gartner-Ellis Theorem for proving uniform large deviation principle for a collection of such sequences. The conditions imposed on the random variables restrict the set of admissible probability laws, but are sufficiently flexible to apply to a wide variety of situations.

\item We provide strong laws characterizing the rate of growth of two performance measures of service under the cloud computing architecture, namely the minimum time taken to observe a continued latency period of a given length, and its dual the maximum latency period that is observed within a given time. 

\item We show, using a motivating example, how these results can be used by a cloud service manager to a) create SLA contracts representing a guarantee to the customer against chances of observing frequent long latencies, and b) design system improvements to minimize the frequency of long latencies, such as rates at which new capacity should be procured/allocated to maintain or improve service.
\end{enumerate}

The following section describes our model of the cloud environment and states the main result of this paper. We conclude the section with a discussion of a representative example. Section~\ref{sec:uldp} states and proves the uniform large deviation principle for collections of random variables. This is used in Section~\ref{sec:lss} where the main result is proved.

\end{section}

\begin{section}{Cloud Model and Main Result}\label{sec:mainresult}

We model the workload of each user with respect to the instantaneous requirements for a single resource, e.g. CPU cycles required, over time. A total of $ K$ customer groups are served, where groups differ in their workload characterization. The cloud is managed in a manner that provisions $ n_{ i}(t)$ customers from the $i$th group at time $ t$ on each large logical server. The function $n_i(t)$ is assumed to be a power function, i.e. there exist a positive constant $ \alpha $ and positive integers $ c_{ 1},\ldots,c_{ K}$, such that 
        \[ 
 	n_{ i}(t)= c_{ i}\lfloor t^{ \alpha }\rfloor \ \ \ \ \mbox{ for all } i=1,\ldots,K.   
\]
For any $ x\in \mathbb{R}$, $ \lfloor x \rfloor$ denotes the greatest integer less than or equal to $ x$ and $ \lceil x\rceil$ represents the smallest integer greater than or equal to $ x$.  The $ c_{ i}$ are chosen to be positive integers rather than real numbers. This is solely because of convenience in handling the limit identities which appear below; we are certain that taking $ n_{ i}(t)=\lfloor c_{ i}t^{ \alpha } \rfloor$ for some positive real number $ c_{ i}$ would not have any significant effect on the results. This form for $n_i(t)$ has two important implications: first, the relative mix of customers from each group, defined by the ratios of the parameters $c_i$, remains constant over time, and only the total population of users grows with time. Second, the number of customers remain a deterministic function of time. We believe this setting can be easily generalized to allow the number of customers to be a stochastic process, e.g. the case where $ (n_{ 1}(t),\ldots,n_{ K}(t))$ are jointly regularly varying with index $ \alpha $ and the number of customers in the $ i$th group  is a Poisson process with intensity $ n_{ i}(t)$, but we do not foresee this situation adding any extra insights to the studied problem.

The $ j$th customer in the $ i$th group has workload $ W_{ i,j}(t)$ at discrete-time $ t$:
        \[ 
 	W_{ i,j}(t)=\mu_{ i}+ X_{ i,j}(t)=\mu_{ i}+\beta_{ i}^{ T} Z(t)+\varepsilon_{ i,j}(t)  \ \ \ \ \mbox{ for all }1\le i\le K, 1\le j\le n_{ i}(t),t\ge 1,
\]
where $ \mu_{i}$ is a constant denoting the expected workload of customers in the $ i$th group and $ X_{ i,j}(t)$ is the deviation from the mean workload of the $ j$th customer in the $ i$th group at time $ t$. The stochastic process $X_{i,j}(t)$ is further defined as the weighed sum of a $K$-dimensional moving-average process $Z(t)$ and an additional pure-noise i.i.d. random variables $\varepsilon_{i,j}(t)$. The weights $\beta_i\in \mathbb{R}^{ K}$ are group-specific constants. 
The noise-process $ (\varepsilon_{ i,j}(t);1\le i\le K, t\ge 1, 1\le j\le n_{ i}(t))$ consists of  independent and identically distributed (i.i.d.) random variables, independent of $(\xi(t),t\in \mathbb{Z}) $, with mean zero, satisfying
\[ 
 	\Lambda _{ \varepsilon }(\lambda ):= \log E \big[  \exp \big\{\lambda \varepsilon _{ i,j}(t)\big\} \big]< \infty \ \ \ \ \mbox{ in a neighborhood of }0.   
\]

The process $Z(t)$ is a $ K$ dimensional moving average process defined as
\[ 
 	Z(t)= \sum_{ k}\phi_{ k}\xi(t-k) \ \ \ \ \mbox{ for all }t\in \mathbb{Z},  
\]
with $ \sum_{ k}|\phi_{ k}|< \infty$. We will assume $ \phi:= \sum_{ k}\phi_{ k}\ne 0$. The {\em innovations} $ (\xi(t);t\in \mathbb{Z})$ are $ K$-dimensional i.i.d. 
random variables with mean zero, satisfying
\begin{equation}\label{eq:def:Lambdaxi} 
 	\Lambda_{ \xi}(\eta ):= \log E\big[ \exp\big\{   \eta \cdot \xi(t)  \big\} \big]< \infty     \ \ \ \ \mbox{ for all } \eta \in \mathbb{R}^{ K},
\end{equation} 
where for any two vectors $ x$ and $ y$, $ x\cdot y$ denotes the scalar product.
We shall place the following additional restriction on the log-m.g.f. $\Lambda_{\xi}(\cdot)$ to satisfy the conditions of the uniform large deviation principle (Theorem~\ref{thm:uldp}):
\begin{assumption}\label{assmp:steep} 
 $  \big| \frac{ d}{d\lambda  }  \Lambda _{ \xi}(\lambda \bar \beta ) \big|\to \infty$ whenever $  | \lambda  |\to \infty$, where $ \bar \beta :=C^{ -1}(\sum_{ i=1}^{ K} c_{ i}   \beta_{ i})$ with $ C:=\sum_{ i=1}^{ K} c_{ i}$.
\end{assumption}

This mild restriction on the parameters of the MA process is satisfied by realistic computing workloads. For example, it admits a Gaussian form for the innovations $\xi$.

The expected workload of the server at time $ t$ is given by $\sum_{ i=1}^{ K}  n_{ i}(t)\mu_{ i}.$
In our setup the number of customers in each group grows over time and so does the expected workload of the server. Hence, to keep the system solvent and avoid build up of an infinite queue, the capacity of the server must also be continually increased. This can be done, for example, by ensuring that the capacity grows in order to maintain a constant ratio of $C_{p}$ with the total expected workload.
Our imperative is to understand the deviations from the mean workload. Define $ S(t)$ as the sum of all the deviations until time $ t$:
 \[ S(t):=\sum_{ k=1}^{ t }\sum_{ i=1}^{ K}\sum_{ j=1}^{ n_{ i}(k)} X_{ i,j}(t) \ \ \ \ \mbox{ for all }t\ge 1.\]
 and  $ N(t)$ as the associated normalizing term for time $ t$:
 \[ 
 	N(t)=\sum_{ k=1}^{ t }\sum_{ i=1}^{ K}   n_{ i}(k) \ \ \ \ \mbox{ for all }t\ge 1.
\]
By convention, we understand that $ \sum_{ l=i}^{ j}x_{ l}=0$ if $ j\le i$. Furthermore, if $ i$ and $ j$ are not integers $ \sum_{ l=i}^{ j}x_{ l}$ will denote $ \sum_{ l=\lceil i\rceil}^{ \lfloor j\rfloor}x_{ l}$. 

We study the average deviation of the workload of the server from its mean over long segments of time. For any time segment $ (k,l)$ the average deviation is given by 
\[ 
 	\bar X(k,l):=\frac{ S(l)-S(k)}{N(l)-N(k)}   .
\]
A  simple argument using law of large numbers tell us that $ \bar X(k,l)$ should not be too far away from $ 0$ if $ l-k$ is large. If $ \bar X(k,l)$ is not close to $ 0$ then we term $ (k,l)$ as a {\em strange segment}. It is also easy to see that if we fix any number $ L$ and a threshold $ \epsilon $ and wait sufficiently long, we will almost surely get a segment $ (k,l)$ such that $ l-k\ge L$ and $ \bar X(k,l)>\epsilon $. Our main result describes how the length of these strange segments grow over time. 

For any measurable set $ A$, we define the \emph{long strange segments} as
\begin{equation} \label{eq:str.seg.gen}
R_t(A):=\sup\left\{ m:\, \bar X(l-m,l)\in A \ \text{for some $l=m,\ldots ,t$}\right\},
\end{equation}
and its dual characteristic
\begin{equation}\label{eq:lss:tr}
	T_{ r}(A):= \inf\left\{l:\ \text{there exists}\  k,0\le k\le l-r \mbox{ such that }\bar X(k,l)\in A \right\}.
\end{equation}
 The functional $R_n(A)$ is the maximum length of a segment from the first $n$ observations whose  average is in set  $A$. $T_n(A)$ is the minimum number of observations required to have a segment of length at least $n$, whose average is in the set $A$. It is easy to see that $ R_{ t}(A)$ grows as $ t\to \infty$ and $ T_{ r}(A)$ grows as $ r\to \infty$. Theorem \ref{thm:lss} below describes the rate of growth of these functionals. There is a duality relation between  the rate of growth of these functionals which follows from the fact $ \{   T_{ r}(A)\le m  \}=\{   R_{ m}(A)\ge r  \}$. If the per-customer capacity of the server is maintained at $ C_{ p}$ units above its expected value then we will take $ A=(C_{ p},\infty)$. 

For any convex function $ f  (\cdot)$, we will use $ f ^{ *}(\cdot)$ to denote its Fenchel-Legendre transform:
\[ 
 	f  ^{ *}(x):=\sup_{ \lambda \in \mathbb{R}}\big\{ \lambda x-f (\lambda ) \big\}    .
\]
For any set $ A\subset \mathbb{R}$, $ A^{ \circ}$ and $ \bar A$ will represent the interior and closure of $ A$ respectively. 

\begin{theorem}\label{thm:lss} 
 For any measurable set $ A$
 \begin{equation}\label{eq:Trresult}
 	 I_{ *}\le \liminf_{ r \rightarrow \infty  } \frac{ \log T_{ r}(A)}{ r }\le     \limsup_{ r \rightarrow \infty  } \frac{\log T_{ r}(A)}{ r }\le I^{ *} \ \ \ \ a.s.  , 
\end{equation}
and
 \begin{equation}\label{eq:Rtresult} 
 	 \frac{ 1}{I^{ *} }\le \liminf_{ t \rightarrow \infty  } \frac{ R_{ t}(A)}{ \log t }\le     \limsup_{ t \rightarrow \infty  } \frac{ R_{ t}(A)}{ \log t }\le \frac{ 1}{I_{ *} } \ \ \ \ a.s.,
\end{equation}
where 
\[ 
 	I_{ *} = \inf_{ x\in \bar A}\Lambda ^{ *}(x) \ \ \mbox{ and } \ \ I^{ *}=\inf_{ x\in A^{ \circ}} \Lambda ^{ *}(x),   
\]
$ \Lambda ^{ *}(x)$  is the Fenchel-Legendre transform of $ \Lambda(\lambda ):=\Lambda _{ \xi}(\lambda \phi\bar \beta ) $. 
\end{theorem}

\begin{remark}
Under our assumption that the customer group compositions remain constant, the customer groups are all jointly represented by their average $\bar \beta = (\sum_{ i=1}^{ K} c_{ i}   \beta_{ i}) / \sum_{ i=1}^{ K} c_{ i}$.
\end{remark}

\begin{remark}\label{rem:fixed capacity} 
We are interested in sets of the nature $ A=(C_{ p},\infty)$, where system-stability requires that the value $C_p$ be set greater than 0. Then, the continuity and increasing nature of the Fenchel-Legendre transform over $A$ ensures that the infimum over the sets $\bar A$ and $A^{\circ}$ are achieved at $C_{p}$. Thus, the upper and lower bounds in \eqref{eq:Trresult} and \eqref{eq:Rtresult}  collapse to give a limit result of the form: 
\begin{equation}\label{eq:limresult}
  \lim_{ r \rightarrow \infty  } \frac{ \log T_{ r}(A)}{ r } = \lim_{ t \rightarrow \infty  } \frac{ \log t }{ R_{ t}(A)} = \Lambda^{*} (C_p) \ \ \ \ a.s.   
\end{equation}
\end{remark}

\begin{example}
Suppose that the innovation vectors $\xi(t)$ are i.i.d. replicates of a $K-$dimensional joint-normal random vector with mean zero and covariance matrix $\Sigma$. In that case  $ \Lambda (\lambda )= \lambda ^{ 2}\phi^{ 2} \bar \beta ^{ T} \Sigma \bar \beta/2$ and hence 
\[ 
 	\Lambda ^{ *}(x) = \big(2 \phi^{ 2} \bar \beta ^{ T} \Sigma \bar \beta     \big)^{ -1} x^{ 2} \ \ \ \ \mbox{ for all } x\in \mathbb{R}.
\]
Therefore, if $ A=(C_{ p},\infty)$ then 
 \begin{equation}\label{eq:limresult:ex}
  \lim_{ r \rightarrow \infty  } \frac{ \log T_{ r}(A)}{ r } = \lim_{ t \rightarrow \infty  } \frac{ \log t }{ R_{ t}(A)}=  \big(2 \phi^{ 2} \bar \beta ^{ T} \Sigma \bar \beta     \big)^{ -1} C_{ p}^{ 2} \ \ \ \ a.s.   
\end{equation}
This yields the estimates $T_r\sim \exp\{rC^2_p/M\}$ and $R_t\sim M\log t/C^2_p$, where $C_p$ represents the server's capacity and $M=2 \phi^{ 2} \bar \beta ^{ T} \Sigma \bar \beta   $ is a property of the customer classes. As expected, higher values of $ C_{ p}$ slow the rate of growth of the duration $T_r$ before observing a latency period of length $r$. On the other hand, higher variability of the innovation $ \xi(t)$ or a higher value of $|\phi|$ in the MA process results in a higher value of  $M$ and culminates in a faster growth of the long latency periods $R_t$ observed in time $t$. 

Another interesting application is when $A=(-\infty,-C_p)$. This can be used to check if there are long time periods when the server resources are being severely under utilized. By the symmetry of the Gaussian distribution, the estimates for $T_r$ and $R_t$ remain the same in this case. In particular, if $C_p$ were chosen equal to the average workload size, then $R_t$ estimates the longest period by time $t$ when the server idles.
\end{example}

We postpone the proof of Theorem \ref{thm:lss} till Section \ref{sec:lss}, and develop the proper tools required for the proof in Section \ref{sec:uldp}. We close this section with a discussion on why standard large-deviation tools are inadequate for the proof of Theorem~\ref{thm:lss}. 

The rate of growth of long strange segments have been studied by \cite{mansfield:rachev:samorodnitsky:2001} for moving average processes with heavy-tailed innovations and then by \cite{rachev:samorodnitsky:2001} for a long-range dependent moving average processes with heavy-tailed innovations. Recently \cite{ghosh:samorodnitsky:2010} studied the effect of memory on the rate of growth of long strange segments for a moving average process with light-tailed innovations. A strong law of the form \eqref{eq:Rtresult} is often referred to as the Erd\"os-R\'enyi law of large numbers; \cite{ErdRen:1970kx} proved asymptotics for longest head runs in i.i.d. coin tosses.

It is instructive to take a heuristic look at the standard technique of proving the rate of growth of long strange segments for a stationary process, say $ (Y_{ t})$.  A vital tool for analyzing this growth is a  large deviation principle associated with the partial sums of $(Y_t)$. Recall that a sequence of probability measures $(P_t,t\ge 1)$ satisfies \emph{large deviation principle} (LDP) on $\mathbb{R}$ if there exists a non-negative lower-semicontinuous function $ I(\cdot)$ such that for any measurable $ A\subset \mathbb{R}$
\begin{equation}\label{eq:LDP}
-\inf\limits_{x\in A^\circ}I(x) \leq \liminf_{t\rightarrow\infty}\frac{1}{t}\log P_t(A) \leq \limsup_{t\rightarrow\infty}\frac{1}{t}\log
P_t(A) \leq -\inf\limits_{x\in \bar{A}}I(x),
\end{equation}
The function $ I(\cdot)$ is called the \emph{rate function}. A rate function with compact level sets is called a \emph{good} rate function. 

Denote the average of the segment $ (k,l)$ by
\[ 
 	\bar Y(k,l)=\frac{ \sum_{ i=k+1}^{ l} Y_{ i}}{ l-k }   .
\]
It is often possible to show that the law of $ \bar Y(0,t)$ satisfies an LDP under assumptions of mixing or other specific structure on $(Y_{ t})$  and existence of exponential moments of $ Y_{ t}$; see for example \cite{bryc:dembo:1996}, \cite{varadhan:1984}, \cite{dembo:zeitouni:1998}, \cite{deuschel:stroock:1989}. Then for a `nice' set $ A$ such that $ E(Y_{ 0})\notin \bar A$  there exists  $ I> 0$ such that for $ t$ large 
\[ 
 	  \log P\big[ \bar Y(0,t) \in A \big] \sim -I t.
\]
Using stationarity, this implies $  \log P\big[ \bar Y(l,l+t) \in A \big] \sim -I t $ for every $ l\ge 0$. Heuristically, this means that for approximately $ e^{ tI}$ segments of length $ t$, we can expect to find one with an average would be in $ A$. The segments $ (0,t),(1,t+1),(2,t+2),\ldots$ are not independent but that is handled typically using mixing type conditions borrowed from the process $ (Y_{ t})$ itself. Theorem 2.3 in \cite{ghosh:samorodnitsky:2010} is an  example of this line of argument where the authors consider moving average processes and use the large deviation principle for partial sums proved in \cite{ghosh:samorodnitsky:2009} to obtain asymptotic results for the rate of growth of  long strange segments.

In our application's setting, the distribution of $ \bar X(l,l+t)$ differs from that of $ \bar X(0,t)$  when  $ l>0$. This is because the growing number of customers in the system implies that each $ \bar X(l,l+t)$ represents an average over different number of realizations ($N(t+l)-N(l)$ versus $N(t)$). So, in order to understand the rate of growth of the long strange segments we need to estimate the probability $ P\big[ \bar X(l,l+t) \in A \big] $ uniformly over $ l\ge 0$. We address this problem by proving the uniform large deviation principle in Theorem~\ref{thm:uldp}. A collection of probability measures $(P_{k,t},t\ge 1, k\in \Gamma )$ satisfies large deviation principle on $\mathbb{R}$ {\em uniformly over} $ k\in \Gamma $ if there exist  non-negative lower-semicontinuous functions $ (I_{ k}(\cdot),k\in \Gamma )$ such that for any measurable $ A\subset \mathbb{R}$
\begin{equation}\label{eq:uldp:defn:lower}
	\liminf_{ t \rightarrow \infty  } \inf_{ k \in \Gamma } \left\{\frac{ 1}{ t}\log P_{ k,t}(A) +\inf_{ x\in  A^{ \circ}}I_{ k}(x)\right\}\ge 0	
\end{equation}
and
\begin{equation}\label{eq:uldp:defn:upper}
	\limsup_{ t \rightarrow \infty  } \sup_{ k\in \Gamma } \left\{\frac{ 1}{ t}\log P_{ k,t}(A) +\inf_{ x\in \bar A}I_{ k}(x)\right\}\le 0.
\end{equation}
Note that bounds \eqref{eq:uldp:defn:lower} and \eqref{eq:uldp:defn:upper} are generalizations of the LHS and RHS of the standard large-deviation bounds in \eqref{eq:LDP}.

\end{section}

\begin{section}{Uniform Large Deviation Principle} \label{sec:uldp}
The Gartner-Ellis Theorem is an important tool for proving large deviation principle, cf. \cite{gartner:1977}, \cite{ellis:1984} and \citep[Theorem 2.3.6, p.44]{dembo:zeitouni:1998}. Theorem \ref{thm:gen:uldp} is an analog of the Gartner-Ellis Theorem for proving uniform large deviation principle. We use this theorem to prove uniform large deviation principle for the average of segments of the server workload process in Theorem \ref{thm:uldp} which is in fact the first step in proving of Theorem \ref{thm:lss}.
\begin{theorem}\label{thm:gen:uldp} 
 Suppose $ (Y_{ k,t},t\ge 1, k\in \Gamma )$ is a collection of random variables such that there exists $ (\Lambda ^{ k}(\cdot),k\in \Gamma )$ which  are differentiable  and satisfy the following conditions: for all $ 0<L<\infty$ and $ \epsilon >0$ there exists $ T>0$ and $ \delta >0$ such that 
 \begin{align}
  	&\lim_{ t \rightarrow \infty  }  \sup_{ k\in \Gamma, | \lambda  |\le L } \left| \Lambda ^{ k}\big(\lambda\big ) - \frac{ 1}{ t} \log E\Big[ \exp \Big\{   t\lambda Y_{ k,t} \Big \}\Big]\right|=0,  \label{eq:cond1}  \\
	& 	\sup_{ k\in \Gamma ,t\ge T,  | \lambda  |\le L} \left| \frac{ 1}{t }  \log E\big[ \exp \{   t\lambda Y_{ k,t}\big] \right|  < \infty ,\label{eq:cond2} \\
	 &\inf_{ k\in \Gamma }\left| (\Lambda ^{ k})^{ \prime }(\lambda ) \right|\to \infty \ \ \ \ \mbox{ whenever }  | \lambda  |\to \infty, \label{eq:cond3}\\
\intertext{and}
	&\left| (\Lambda ^{ k})^{ \prime }(\lambda _{ 1})-(\Lambda ^{ k})^{ \prime }(\lambda _{ 2}) \right|<\epsilon \ \ \ \ \mbox{ for all }|\lambda _{ 1}-\lambda _{ 2}|<\delta, \lambda_{ 1},\lambda _{ 2}\in [-L,L],k\in \Gamma .\label{eq:cond:uniform:cont}	 
\end{align}
Then for any  closed set $ F\subset \mathbb{R}$
\begin{equation}\label{eq:upper:limit:gen}
	\limsup_{ t \rightarrow \infty  } \sup_{ k\in \Gamma } \left\{\frac{ 1}{ t}\log P \left[ Y_{ k,t}\in F \right] +\inf_{ x\in F}\Lambda ^{ k*}(x)\right\}\le 0
\end{equation}
and for any  open set $ G\subset \mathbb{R}$
\begin{equation}\label{eq:lower:bound:gen}
	\liminf_{ t \rightarrow \infty  } \inf_{ k \in \Gamma } \left\{\frac{ 1}{ t}\log P \left[ Y_{ k,t}\in G\right] +\inf_{ x\in  G}\Lambda ^{ k*}(x)\right\}\ge 0	
\end{equation}
where the rate function $ \Lambda ^{ k*}(\cdot)$ is the Fenchel-Legendre transform of $ \Lambda ^{ k}(\cdot)$.
\end{theorem}

\begin{remark}\label{rem:thm3:1:comments} 
It can be observed from the proof below that conditions \eqref{eq:cond1}, \eqref{eq:cond2} and \eqref{eq:cond3} have been used to prove  \eqref{eq:upper:limit:gen}, whereas, all the conditions \eqref{eq:cond1}-\eqref{eq:cond:uniform:cont} are required for proving \eqref{eq:lower:bound:gen}. Condition \eqref{eq:cond1} requires that the normalized log-m.g.f.s of $Y_{ k,t}$ converges to $\Lambda ^{ k}(\lambda )$ uniformly over $ k\in \Gamma $ and locally uniformly in $ \lambda \in \mathbb{R}$. Condition \eqref{eq:cond2} ensures uniform exponential tightness of the random variables $ (Y_{ k,t})$. Condition \eqref{eq:cond3} is the equivalent of the steepness assumption imposed by the Gartner-Ellis theorem, cf. \citep[Theorem 2.3.6, p.44]{dembo:zeitouni:1998}. Condition \eqref{eq:cond:uniform:cont} requires that the functions $ (\Lambda ^{ k })^{ \prime }(\lambda )$ are continuous  in $ \lambda $, uniformly over $ k\in \Gamma $ and $ \lambda $ in a compact subset of $ \mathbb{R}$. This ensures that the Fenchel-Legendre transforms $ \Lambda ^{ k*}(x)$ are continuous in $ x$, uniformly over $ k\in \Gamma $ and $ x$ in compact subsets of $ \mathbb{R}$. 
\end{remark}

\begin{proof} 
We will first prove \eqref{eq:upper:limit:gen}. As \eqref{eq:upper:limit:gen} holds trivially when $ F=\emptyset$, we can safely  assume that $ F$ is non-empty. To begin with suppose  $ F$ is compact.  Fix any $ x\in F$ and $ \delta >0$. Since $ \Lambda ^{ k}(\cdot)$ is convex, continuously differentiable and satisfies \eqref{eq:cond3}, we can find  $ \lambda ^{ k}_{ x}\in \mathbb{R}$ such that 
\[ 
 	\big( \Lambda ^{ k} \big) ^{ \prime }(\lambda^{ k}_{ x})=x.
\]
This would imply
\[ 
 	\Lambda ^{ k*}(x)=\sup_{ \lambda \in \mathbb{R}}\{   \lambda x-\Lambda ^{ k}(\lambda )  \}=\lambda ^{ k}_{ x}x-\Lambda ^{ k}(\lambda^{ k}_{ x} ). 
\]
From \eqref{eq:cond3} we also know  that $ \{   \lambda ^{ k}_{ x} :k\in \Gamma  \}$ is a bounded set. Hence we can find an open neighborhood $ A_{ x}$ of $ x$ such that 
\[ 
 	\inf_{ y\in A_{ x}} \lambda ^{ k}_{ x}(y-x)\ge -\delta \ \ \ \ \mbox{ for all } k\in \Gamma.   
\] 
Then by Chebychev's inequality we  get an upper bound for the following probability
\[
	P\big[ Y_{ k,t}\in A_{ x} \big]    \le E\Big[ \exp \Big\{ \lambda ^{ k}_{ x}t\big(Y_{ k,t}-x\big) \Big\}  \Big] \exp\Big\{- t\inf_{ y\in A_{ x}} \lambda ^{ k}_{ x}(y-x) \Big\} 
\]
which implies
\[
	 \frac{ 1}{t }\log P\big[ Y_{ k,t}\in A_{ x} \big]     \le \frac{ 1}{ t} \log E \Big[ \exp \Big\{ \lambda ^{ k}_{ x}tY_{ k,t} \Big\}  \Big]-\lambda ^{ k}_{ x}x+\delta .
\]
From \eqref{eq:cond1} we can get $ T\ge 1$ such that for all $ t\ge T$ and $ k\in \Gamma $
\[ 
 	\frac{ 1}{ t} \log E \Big[ \exp \Big\{ \lambda ^{ k}_{ x}tY_{ k,t} \Big\}  \Big]\le \Lambda ^{ k}\big( \lambda ^{ k}_{ x} \big) +\delta 
\]
and this means for $ t\ge T$
\begin{equation}\label{eq:small:ball}
	 \frac{ 1}{t }\log P\big[ Y_{ k,t}\in A_{ x} \big]     \le\Lambda ^{ k}\big( \lambda ^{ k}_{ x} \big)-\lambda ^{ k}_{ x}x+2\delta=-\Lambda ^{ k*}(x)+2\delta  .
\end{equation}
Now, obviously $ \cup_{ x\in F}A_{ x}$ is an open  cover of $ F$ and since $ F$ is compact we can obtain $ x_{ 1},\ldots,x_{ N}\in F$, such that 
$ F\subset \cup_{ 1\le i\le N} A_{ x_{ i}}$. Then by a simple union of events bound we get for $ t\ge T$
\[ 
 	  \frac{ 1}{t }\log P\big[ Y_{ k,t}\in F \big]+\min_{ 1\le i\le N} \Lambda ^{ k*}(x_{ i}) \le \frac{ 1}{t }\log N+2\delta \ \ \ \ \mbox{ for all }k\in \Gamma .
\]
It is now easy to see that for $ t\ge T$
\[ 
 	   \sup_{ k\in \Gamma }\left\{\frac{ 1}{t }\log P\big[ Y_{ k,t}\in F \big]+\inf_{ x\in F} \Lambda ^{ k*}(x)\right\} \le \frac{ 1}{t }\log N+2\delta 
\]
and since $ 0<\delta<1 $ is arbitrary 
\begin{equation}\label{eq:upper:bound:compact} 
 	\limsup_{ t \rightarrow\infty  }      \sup_{ k\in \Gamma }\left\{\frac{ 1}{t }\log P\big[ Y_{ k,t}\in F \big]+\inf_{ x\in F} \Lambda ^{ k*}(x)\right\}\le 0.
\end{equation}
This proves \eqref{eq:upper:limit:gen} when $ F$ is compact.

Next we extend the above result to any non-empty closed set $ F$. First we note a few facts. Using  \eqref{eq:cond1} and \eqref{eq:cond2}  we get that for any $ \delta >0$
\[ 
 	c:=\sup_{ k\in \Gamma,  | \lambda  |<\delta }   \big| \Lambda ^{ k}(\lambda ) \big| <\infty.
\]
Since $ \{   \lambda ^{ k}_{ x}:k\in \Gamma   \}$ is bounded and $ \Lambda ^{ k*}(x)=\lambda ^{ k}_{ x}x-\Lambda ^{ k}(\lambda ^{ k}_{ x})$ we get that $ \sup_{ k\in \Gamma }\Lambda ^{ k*}(x)< \infty$.
Furthermore, for all $ k\in \Gamma $
\[ 
 	\Lambda ^{ k*}(x)=\sup_{ \lambda \in \mathbb{R}}\big\{ \lambda x-\Lambda ^{ k}(\lambda ) \big\}    \ge \sup_{| \lambda |< \delta }\big\{ \lambda x-\Lambda ^{ k}(\lambda ) \big\}\ge \delta  | x |-c.
\] 
Hence for any closed set $ F$ there exists $ M_{ 1}>0$ such that 
\begin{equation}\label{eq:uniform:good:ratefn} 
 	\inf_{ x\in F}\Lambda ^{ k*}(x)=\inf_{ x\in F\cap[-M_{ 1},M_{ 1}]} \Lambda ^{ k*}(x)   \ \ \ \ \mbox{ for all }k\in \Gamma .
\end{equation}
Also note that for any $ k\in \Gamma $ and $ t\ge 1$
\[ 
 	\frac{ 1}{t }\log P\big[ |Y_{ k,t}|>\theta  \big] \le -\theta +\sup_{ k\in \Gamma ,t\ge 1}\frac{ 1}{t }\log E\Big[ e^{ tY_{ k,t}} \Big]    +\sup_{ k\in \Gamma,t\ge 1 }\frac{ 1}{t }\log E\Big[ e^{- tY_{ k,t}} \Big]  
\]
and therefore
\begin{equation}\label{eq:exp:tight} 
 	\lim_{ \theta  \rightarrow \infty  } \limsup_{ t \rightarrow\infty  } \sup_{ k\in \Gamma }    \frac{ 1}{t }\log P\big[ |Y_{ k,t}|>\theta  \big]=-\infty.
\end{equation}
Now set 
\[ 
 	c^{ \prime }=\sup_{ k\in \Gamma }\inf_{ x\in F} \Lambda ^{ k*}(x)   
\]
Since for any $ x$, $ \sup_{ k\in \Gamma }\Lambda ^{ k*}(x)< \infty$ we get that $ c^{ \prime }<\infty$. Note that if $ c^{ \prime }=0$ then the proof is immediate. So we look into the case when $ c^{ \prime }>0$. Using \eqref{eq:exp:tight} we can  get $ M_{ 2}>0$ such that 
\[ 
 	P\big[| Y_{ k,t}|>M_{ 2} \big]\le e^{ -2c^{ \prime }t} \ \ \ \ \mbox{ for all }k\in \Gamma ,t\ge 1.    
\]
Let $ M=\max\{   M_{ 1},M_{ 2}  \}$. Note that from \eqref{eq:upper:bound:compact} and \eqref{eq:uniform:good:ratefn} 
\begin{align*}
	&    	\limsup_{ t \rightarrow\infty  }      \sup_{ k\in \Gamma }\left\{\frac{ 1}{t }\log P\big[ Y_{ k,t}\in F\cap[-M,M] \big]+\inf_{ x\in F} \Lambda ^{ k*}(x)\right\}\\
	&  =  	\limsup_{ t \rightarrow\infty  }      \sup_{ k\in \Gamma }\left\{\frac{ 1}{t }\log P\big[ Y_{ k,t}\in F\cap[-M,M] \big]+\inf_{ x\in F\cap[-M,M]} \Lambda ^{ k*}(x)\right\}\le 0.
 \end{align*}
 This means that for any given $ \delta >0$ we can find $ T\ge 1$ such that 
 \[ 
 	 \frac{ 1}{t }\log P\big[ Y_{ k,t}\in F\cap[-M,M] \big]+\inf_{ x\in F} \Lambda ^{ k*}(x)\le \delta \ \ \ \ \mbox{ for all }k\in \Gamma ,t\ge T.  
\]
Now if $ P[Y_{ k,t}\in F\cap[-M,M] ]\le P[  | Y_{ k,t} |>M]$ then 
\[ 
 	\frac{ 1}{t }\log P\big[ Y_{ k,t}\in F \big]   \le \frac{ 1}{t }\log 2-2c^{ \prime }.
\]
Otherwise,
\[ 
 	\frac{ 1}{t }\log P\big[ Y_{ k,t}\in F \big]   \le \frac{ 1}{t }\log 2+\frac{ 1}{t }\log P\big[ Y_{ k,t}\in F\cap[-M,M] \big] .
\]
Therefore, in both the cases,
 \[ 
 	 \frac{ 1}{t }\log P\big[ Y_{ k,t}\in F \big]+\inf_{ x\in F} \Lambda ^{ k*}(x)\le \frac{ 1}{t }\log 2+\delta \ \ \ \ \mbox{ for all }k\in \Gamma ,t\ge T.  
\]
and hence
\[ 
 	\limsup_{ t \rightarrow\infty  }      \sup_{ k\in \Gamma }\left\{\frac{ 1}{t }\log P\big[ Y_{ k,t}\in F \big]+\inf_{ x\in F} \Lambda ^{ k*}(x)\right\}\le 0.
\]
This completes the proof of \eqref{eq:upper:limit:gen}.

We will now prove \eqref{eq:lower:bound:gen}.   Note that we  can find $ M>0$ such that 
\[ 
 	   \inf_{ x\in G}\Lambda ^{ k*}(x)=\inf_{ x\in G\cap[-M,M]}\Lambda ^{ k*}(x) \ \ \ \ \mbox{ for all }k\in \Gamma  .
\]
Fix any $ \epsilon >0$ and get $ x^{ k}\in G\cap[-M,M]$ such that 
\[ 
 	\Lambda ^{ k*}(x^{ k})  <\inf_{ x\in G}\Lambda ^{ k*}(x)+\epsilon /2. 
\]
Another observation that we need to make is that we can find $ \delta >0$ such that  
\[ 
	\left| \Lambda ^{ k*}(x)-\Lambda ^{ k*}(y) \right|<\epsilon /2 \ \ \ \ \mbox{ for all }| x-y |<\delta,x,y\in [-M,M], k\in \Gamma .
\]
This follows easily from \eqref{eq:cond:uniform:cont}. Now obviously $ \cup_{ x\in G\cap[-M,M]}B_{ x,\delta }$ is an open cover of $ G\cap[-M,M]$, where $ B_{ x,\delta }=(x-\delta ,x+\delta )$. Since $ G\cap[-M,M]$ is precompact, we can find $ x_{ 1},\ldots,x_{ n}\in G\cap[-M,M]$ such that for all $ x^{ k}$ there exists $ 1\le i_{ k}\le n$ for which $  | x^{ k}-x_{ i_{ k}} |<\delta $. This implies that 
\[ 
 	\inf_{ 1\le i\le n} \Lambda ^{ k*}(x_{ i})<\inf_{ x\in G}\Lambda ^{ k*}(x)+\epsilon \ \ \ \ \mbox{ for all }k\in \Gamma .   
\]
For notational simplicity we define $ X=\{   x_{ 1},\ldots,x_{ n}  \}$. Let $ \delta ^{ \prime }>0$ be such that $ B_{ x,\delta^{ \prime } }\subset G$ for all $ x\in X$.  Now fix any $ x\in X$.  Define the random variables $ \tilde Y_{ k,t}$ by an exponential change of measure such that
\[ 
 	P\big[ \tilde Y_{ k,t}\in B \big]    = \frac{ E\big[ e^{ t\lambda ^{ k}_{ x}Y_{ k,t}}I_{ [Y_{ k,t}\in B]} \big] }{E\big[ e^{ t\lambda ^{ k}_{ x}Y_{ k,t}} \big] }
\]
Then 
\begin{align*}
	P\big[ Y_{ k,t}\in B_{ x,\delta^{ \prime } } \big] &  = E\big[ e^{ t\lambda ^{ k}_{ x}Y_{ k,t}} \big] E\Big[ e^{ -t\lambda ^{ k}_{ x}\tilde Y_{ k,t}}I_{ [\tilde Y_{ k,t}\in B_{ x,\delta^{ \prime } }]} \Big]   
 \end{align*}
 and
 \begin{align*} 
 	\frac{ 1}{t }\log   P\big[ Y_{ k,t}\in B_{ x,\delta^{ \prime } } \big]&= \frac{ 1}{t }\log  E\big[ e^{ t\lambda ^{ k}_{ x}Y_{ k,t}} \big]+ \frac{ 1}{t }\log E\Big[ e^{ -t\lambda ^{ k}_{ x}\tilde Y_{ k,t}}I_{ [\tilde Y_{ k,t}\in B_{ x,\delta^{ \prime } }]} \Big]\\
	& \ge \frac{ 1}{t }\log  E\big[ e^{ t\lambda ^{ k}_{ x}Y_{ k,t}} \big]-\lambda ^{ k}_{ x}x-  \left| \lambda ^{ k}_{ x} \right|\delta^{ \prime } + \frac{ 1}{t }\log P\big[ \tilde Y_{ k,t}\in B_{ x,\delta^{ \prime } } \big] .
\end{align*}
We claim that 
\begin{equation}\label{eq:exp:changeofmeasure}
	\lim_{ t \rightarrow \infty  }\inf_{ k\in \Gamma ,x\in X} \frac{ 1}{t }\log P\big[ \tilde Y_{ k,t}\in B_{ x,\delta^{ \prime } } \big] =0.
\end{equation}
To remain with the flow we  complete the proof of \eqref{eq:lower:bound:gen} assuming \eqref{eq:exp:changeofmeasure}, which we  prove at the end. Let $ M^{ \prime }>0$ be such that $  | (\Lambda ^{ k})^{ \prime }(\lambda ) |>M$ for all $ |\lambda |>M^{ \prime }$ and $ k\in \Gamma $. From assumption \eqref{eq:cond3} we know that $ M^{ \prime }<\infty$. We can also get $ T\ge 1$ such that for all $ t\ge T$ and $x\in X,$
\[ 
 	   \inf_{ k\in \Gamma } \frac{ 1}{t }\log P\big[ \tilde Y_{ k,t}\in B_{ x,\delta^{ \prime } } \big] \ge -\epsilon  
\]
and 
\[ 
 	\sup_{ k\in \Gamma } \left| \Lambda ^{ k}\big(\lambda^{ k}_{ x}\big ) - \frac{ 1}{ t} \log E\Big[ e^{    t\lambda ^{ k}_{ x}Y_{ t,k} }\Big]\right|<\epsilon .   
\]
This implies for all $ t\ge T,x\in X$ and $ k\in \Gamma $
\[ 
 	 \frac{ 1}{t }\log   P\big[ Y_{ k,t}\in G\big]\ge  \frac{ 1}{t }\log   P\big[ Y_{ k,t}\in B_{ x,\delta^{ \prime } } \big]\ge \Lambda ^{ k}(\lambda ^{ k}_{ x})-\lambda ^{ k}_{ x}x-M^{ \prime }\delta^{ \prime }-2\epsilon =-\Lambda ^{ k*}(x) -M^{ \prime }\delta^{ \prime }-2\epsilon.
\]
Since $ x\in X$ is arbitrary and $ M^{ \prime },\delta ^{ \prime }$ and $ \epsilon $ are independent of the choice of $ x$, we get for all $ t\ge T$ and $ k\in \Gamma $
\[ 
 	 \frac{ 1}{t }\log   P\big[ Y_{ k,t}\in G  \big]\ge-\inf_{ x\in X}\Lambda ^{ k*}(x)-M^{ \prime }\delta^{ \prime }-2\epsilon   \ge -\inf_{ x\in G}\Lambda ^{ k*}(x)-M^{ \prime }\delta^{ \prime }-3\epsilon.   
\]
Hence we get 
\[ 
 	\liminf_{ t \rightarrow\infty  }\sup_{ k\in \Gamma } \left\{\frac{ 1}{t }\log   P\big[ Y_{ k,t}\in G  \big]+\inf_{ x\in G}\Lambda ^{ k*}(x)\right\}\ge  -M^{ \prime }\delta^{ \prime }-3\epsilon .
\]
This completes the proof of \eqref{eq:lower:bound:gen}  since $ \delta^{ \prime }$ and $ \epsilon $ can be chosen arbitrarily close to $ 0$.

It now remains to prove \eqref{eq:exp:changeofmeasure}. Since $ X$ is a finite set, it suffices to show that for any $ x\in X$
\[
	\lim_{ t \rightarrow \infty  }\inf_{ k\in \Gamma } \frac{ 1}{t }\log P\big[ \tilde Y_{ k,t}\in B_{ x,\delta^{ \prime } } \big] =0.
\]
We will use the upper large deviation bound \eqref{eq:upper:limit:gen} for that purpose. Note that 
\begin{align*} 
 	\frac{ 1}{t }\log E\big[ e^{ t\lambda \tilde Y_{ k,t}} \big]   & =\frac{ 1}{t }\log E\big[ e^{ t(\lambda +\lambda ^{ k}_{ x})Y_{ k,t}} \big] - \frac{ 1}{t }\log E\big[ e^{ t\lambda ^{ k}_{ x}Y_{ k,t}} \big] \\
	& \to \tilde \Lambda ^{ k}(\lambda ):= \Lambda ^{ k}(\lambda +\lambda ^{ k}_{ x})-\Lambda ^{ k}(\lambda ^{ k}_{ x}).
\end{align*}
It is easy to check that $ \tilde \Lambda ^{ k}(\cdot)$ inherits the properties \eqref{eq:cond1}, \eqref{eq:cond2}, \eqref{eq:cond3} and \eqref{eq:cond:uniform:cont} from $ \Lambda ^{ k}(\cdot)$.  Therefore, since $ B_{ x,\delta^{ \prime } }^{ c}:=\{   x\in \mathbb{R}:x\notin B_{ x,\delta^{ \prime } }  \}$ is a closed set, by \eqref{eq:upper:limit:gen} 
\begin{equation}\label{eq:upper:tildeY}
	\limsup_{ t \rightarrow\infty  } \sup_{ k\in \Gamma } \left\{  \frac{ 1}{t }\log P\big[ \tilde Y_{ k,t}\in B^{ c}_{ x,\delta^{ \prime } } \big] +\inf_{ y\in B_{ x,\delta ^{ \prime }}^{ c}} \tilde\Lambda ^{ k*}(y) \right\} \ge0.
\end{equation}
Note that $ (\tilde \Lambda ^{ k})^{ \prime }(0)=x$ for all $ k\in \Gamma $ and that implies $ \tilde \Lambda ^{ k*}(x)=0$ for all $ k\in \Gamma $. Since $ \tilde \Lambda ^{ k*}(\cdot)$ is nonnegative and convex $ \inf_{ y\in B^{ c}_{ x,\delta ^{ \prime }}}\tilde \Lambda ^{ k*}(y)\ge\min\{   \tilde\Lambda ^{ k*}(x-\delta ^{ \prime }),\tilde \Lambda ^{ k*}(x+\delta ^{ \prime })  \}$. Now get a compact set $ K^{ \prime }$ such that $ |(\tilde \Lambda ^{ k})^{ \prime }(\lambda )|>|x|+\delta ^{ \prime }$ and then  find $ \eta>0$ such that 
\begin{equation}\label{eq:uniform:cont:tildeLambda} 
 	\left| (\tilde\Lambda ^{ k})^{ \prime }(\lambda^{ \prime }) -(\tilde \Lambda ^{ k})^{ \prime }(\lambda ^{ \prime \prime })  \right|<\delta ^{ \prime  }/2 \ \ \ \ \mbox{ for all } | \lambda ^{ \prime }-\lambda ^{ \prime \prime } |<\eta, \lambda ^{ \prime },\lambda ^{ \prime \prime }\in K^{ \prime }, k\in \Gamma .   
\end{equation}
Then get $\tilde  \lambda ^{ k}_{ x+}$ and $ \tilde \lambda ^{ k}_{ x-}$ such that $ (\tilde \Lambda ^{ k})^{ \prime }(\tilde \lambda ^{ k}_{ x+})=x+\delta ^{ \prime }$ and $ (\tilde \Lambda ^{ k})^{ \prime }(\tilde \lambda ^{ k}_{ x-})=x-\delta ^{ \prime }$. From \eqref{eq:uniform:cont:tildeLambda} we know that $ \tilde \lambda ^{ k}_{ x+}>\eta$ and $ \tilde \lambda ^{ k}_{ x-}<-\eta$ for all $ k\in \Gamma $. 
Therefore, for all $ k\in \Gamma $
\begin{align*}
	\tilde \Lambda ^{ k*}(x+\delta ^{ \prime })&= \tilde \lambda ^{ k}_{ x+}(x+\delta ^{ \prime })-\tilde \Lambda ^{ k}(\tilde \lambda ^{ k}_{ x+})=\tilde \lambda ^{ k}_{ x+}(x+\delta ^{ \prime })-\int_{ 0}^{\tilde \lambda ^{ k}_{ x+} }(\tilde \Lambda ^{ k})^{ \prime }(z)dz\\    
	&  \ge \tilde \lambda ^{ k}_{ x+}(x+\delta ^{ \prime })-(x+\delta ^{ \prime }/2)\eta-(\tilde \lambda ^{ k}_{ x+}-\eta )(x+\delta ^{ \prime })=\eta\delta ^{ \prime }/2,\\
\intertext{and}
	\tilde \Lambda ^{ k*}(x-\delta ^{ \prime })&= \tilde \lambda ^{ k}_{ x-}(x-\delta ^{ \prime })-\tilde \Lambda ^{ k}(\tilde \lambda ^{ k}_{ x-})=\tilde \lambda ^{ k}_{ x+}(x+\delta ^{ \prime })+\int^{ 0}_{\tilde \lambda ^{ k}_{ x-} }(\tilde \Lambda ^{ k})^{ \prime }(z)dz\\    
	&  \ge \tilde \lambda ^{ k}_{ x-}(x-\delta ^{ \prime })+(x-\delta ^{ \prime }/2)\eta+(\tilde \lambda ^{ k}_{ x+}-\eta )(x-\delta ^{ \prime })=\eta\delta ^{ \prime }/2.
 \end{align*}
This implies that $ \min\{   \tilde\Lambda ^{ k*}(x-\delta ^{ \prime }),\tilde \Lambda ^{ k*}(x+\delta ^{ \prime })  \}\ge\eta\delta ^{ \prime }/2$ for all $ k\in \Gamma $ and hence using  \eqref{eq:upper:tildeY} we get
\[ 
 	\limsup_{ t \rightarrow\infty  } \sup_{ k\in \Gamma }\frac{ 1}{t }\log P\big[ \tilde Y_{ k,t}\in B_{ x,\delta ^{ \prime }}^{ c} \big]    \le -\eta\delta ^{ \prime }/2.
\]
This also means that 
\[ 
 	\lim_{ t \rightarrow\infty  }\inf_{ k\in \Gamma }P\big[ \tilde Y_{ k,t}\in B_{ x,\delta ^{ \prime }} \big]    = 1.
\]
This proves \eqref{eq:exp:changeofmeasure} and hence completes the proof of the theorem.
\end{proof}

Theorem \ref{thm:uldp} allows us to approximate the probability of deviation from 0 of the average $ \bar X(k,l)$ for different segments $ (k,l)$ when $ l-k$ is large. This is a vital component in the proof of Theorem \ref{thm:lss}.

\begin{theorem}\label{thm:uldp} 
If Assumption \ref{assmp:steep} holds then for any measurable set $ A\subset \mathbb{R}$
\begin{equation}\label{eq:upper:limit}
	\limsup_{ t \rightarrow \infty  } \sup_{ k\ge 0} \left\{\frac{ 1}{ t}\log P \left[ \bar X\big(kt,(k+1)t\big ) \in A \right] +\inf_{ x\in \bar A}\Lambda ^{ k*}(x)\right\}\le 0
\end{equation}
and 
\begin{equation}\label{eq:lower:bound}
	\liminf_{ t \rightarrow \infty  } \inf_{ k\ge 0} \left\{\frac{ 1}{ t}\log P \left[  \bar X\big(kt,(k+1)t \big) \in A \right] +\inf_{ x\in  A^{ \circ}}\Lambda ^{ k*}(x)\right\}\ge 0	
\end{equation}
where the rate function $ \Lambda ^{ k*}(\cdot)$ is the Fenchel-Legendre transform of 
\begin{equation}\label{eq:limit:LMGF}
	\Lambda ^{ k}(\lambda ):=\int_{ k}^{ k+1} \Lambda _{ \xi} \left( \frac{(\alpha +1)\lambda \phi  y^{ \alpha }}{(k+1)^{ \alpha +1}-k^{ \alpha +1}}\bar\beta   \right)dy,
\end{equation}
and $ \Lambda _{ \xi}(\cdot)$ is as defined in \eqref{eq:def:Lambdaxi}. 
\end{theorem}

\begin{proof} 
The result will follow once we  check that the conditions of Theorem \ref{thm:gen:uldp} hold by setting 
\[ 
 	Y_{ k,t}:=\bar X(kt,(k+1)t)=    \frac{ S((k+1)t)-S(kt)}{N((k+1)t) -N(kt)} \ \ \ \ \mbox{ for all }t\in \mathbb{N}, k\in \mathbb{R}_{ +}.
\]
The most complicated part is to check the uniform convergence condition \eqref{eq:cond1}: for any $ 0<\Delta<\infty$
\begin{equation} \label{eq:unif:conv}
 	\lim_{ t \rightarrow \infty  } \sup_{ k\ge 0,|\lambda |\le \Delta} \left|  \Lambda ^{ k}(\lambda ) - \frac{ 1}{ t} \log E \exp\left\{t \lambda\bar X(kt,(k+1)t) \right\}\right|   =0.
\end{equation}
 We begin  by observing that for any $ u \in \mathbb{R}$
 \begin{align}\label{eq:firsteq}
	& \log E\Big[ \exp \Big\{ u \big( S((k+1)t)-S(kt) \big)  \Big\}  \Big] \nonumber    \\
	&  = \log E\left[ \exp \left\{ u\sum_{ l= kt+1}^{ (k+1)t}\sum_{ i=1}^{ K}\sum_{ j=1}^{ n_{ i}(l)} X_{ i,j}(l) \right\}  \right] \nonumber\\
	& = \log E\left[ \exp \left\{ u \sum_{ l=kt+1}^{ (k+1)t}\sum_{ i=1}^{ K} n_{ i}(l)\beta_{ i}^{ T} Z(l)+ u \sum_{ l=kt+1}^{ (k+1)t}\sum_{ i=1}^{ K}\sum_{ j=1}^{ n_{ i}(l)}\varepsilon_{ i,j}(l)   \right\}  \right]  \nonumber\\
	& = \log E\left[ \exp \left\{ u \sum_{ l=kt+1}^{ (k+1)t}\sum_{ i=1}^{ K} n_{ i}(l)\beta_{ i}^{ T} Z(l)  \right\}  \right]  + \log E\left[ \exp \left\{u\sum_{ l=kt+1}^{ (k+1)t}\sum_{ i=1}^{ K}\sum_{ j=1}^{ n_{ i}(l)}\varepsilon_{ i,j}(l)   \right\}  \right] ,
 \end{align} 
 where the last equality follows from the independence of the $ \varepsilon $'s and the $ Z$'s. 
 To understand the first component of \eqref{eq:firsteq}, define $ \beta =\sum_{ i=1}^{ K}c_{ i}\beta_{ i}$ and note that 
\begin{align*}
	& \log E\left[ \exp \left\{ u \sum_{ l=kt+1}^{ (k+1)t}\sum_{ i=1}^{ K} n_{ i}(l)\beta_{ i}^{ T} Z(l)  \right\}  \right]  \\    
	& = \log E\left[ \exp \left\{ u \sum_{ i=1}^{ K} \beta_{ i}^{ T}\Big(\sum_{ l=kt+1}^{ (k+1)t} n_{ i}(l)\sum_{ j= - \infty}^{  \infty}\phi_{ k}\xi(l-j)  \Big)\right\}  \right] \\ 
	& =\log E\left[ \exp \left\{ u \Big(\sum_{ i=1}^{ K} \beta_{ i}c_{ i}\Big)\cdot \Big(\sum_{ l=kt+1}^{ (k+1)t} \lfloor l^{ \alpha }\rfloor\sum_{ j= - \infty}^{  \infty}\phi_{ j}\xi(l-j)  \Big)\right\}  \right] \\ 
	& = \log E\left[ \exp \left\{ u \beta \cdot \Big(  \sum_{ j=- \infty}^{  \infty} \xi(j) \sum_{ l=kt+1}^{ (k+1)t}\lfloor l^{ \alpha }\rfloor \phi_{ l-j}  \Big)\right\}  \right] \\ 
	&  = \sum_{ j=- \infty}^{  \infty}  \Lambda _{ \xi} \left(u  \beta \sum_{ l=kt+1}^{ (k+1)t} \lfloor l^{ \alpha }\rfloor \phi_{ l-j} \right).
 \end{align*}
Using the triangle inequality we get the obvious bound
\begin{align}\label{eq:triangle}
	&    \lim_{ t \rightarrow \infty  } \sup_{ k\ge 0, | \lambda  |\le \Delta} \left|  \Lambda ^{ k}(\lambda ) - \frac{ 1}{ t} \log E \exp\left\{t \lambda\bar X(kt,(k+1)t) \right\}\right|\\
	&  \le \lim_{ t \rightarrow \infty  } \sup_{ k\ge 0, | \lambda  |\le \Delta} \left|  \Lambda ^{ k}(\lambda ) - \frac{ 1}{ t}  \sum_{ j=kt+1}^{ (k+1)t}  \Lambda _{ \xi} \left(\frac{ t\lambda }{N((k+1)t)-N(kt)  }  \beta \sum_{ l=kt+1}^{ (k+1)t}\lfloor l^{ \alpha }\rfloor \phi_{ l-j} \right)\right|\nonumber\\
	& \ \ + \ \ \lim_{ L \rightarrow \infty  } \lim_{ t \rightarrow \infty  } \sup_{ k\ge 0, | \lambda  |\le \Delta}\left|\frac{ 1}{t } \sum_{ j=- \infty}^{  kt-L}  \Lambda _{ \xi} \left(\frac{ t\lambda }{N((k+1)t)-N(kt)  }  \beta \sum_{ l=kt+1}^{ (k+1)t}\lfloor l^{ \alpha }\rfloor \phi_{ l-j} \right)\right|\nonumber\\
	& \ \ + \ \ \lim_{ L \rightarrow \infty  }\lim_{ t \rightarrow \infty  } \sup_{ k\ge 0, | \lambda  |\le \Delta}\left|\frac{ 1}{t } \sum_{ j=(k+1)t+L}^{ \infty}  \Lambda _{ \xi} \left(\frac{ t\lambda }{N((k+1)t)-N(kt)  }  \beta \sum_{ l=kt+1}^{ (k+1)t}\lfloor l^{ \alpha }\rfloor \phi_{ l-j} \right)\right|\nonumber\\
	& \ \ + \ \ \lim_{ t \rightarrow \infty  }  \sup_{ k\ge 0, | \lambda  |\le \Delta} \left| \frac{ 1}{ t}\sum_{\substack{ kt-L<j\le kt \mbox{ or }\\(k+1)t<j\le (k+1)t+L}}  \Lambda _{ \xi} \left(\frac{ t\lambda }{N((k+1)t)-N(kt)  }  \beta \sum_{ l=kt+1}^{ (k+1)t}\lfloor l^{ \alpha }\rfloor \phi_{ l-j} \right)\right|\nonumber\\
	& \ \  +\ \ \lim_{ t \rightarrow \infty }\sup_{ k\ge 0, | \lambda  |\le \Delta} \left| \frac{ 1}{ t}\log E\left[ \exp \left\{ \frac{t \lambda }{N((k+1)t)-N(kt) }\sum_{ l=1}^{ t}\sum_{ i=1}^{ K}\sum_{ j=1}^{ n_{ i}(l)}\varepsilon_{ i,j}(l)   \right\}  \right] \right|.\nonumber
 \end{align}
 
 We will prove \eqref{eq:unif:conv}  by showing that each of the term in the above expression is equal to 0. For that purpose we make use the following facts:
\begin{enumerate}[(i)] 
\item there exists $ M^{ \prime }>0$ such that 
\[ 
 	\frac{ t((k+1)t)^{ \alpha }}{ (kt+1)^{ \alpha }+\cdots+((k+1)t)^{ \alpha } }\le M^{ \prime } \ \ \ \ \mbox{ for all }t\ge 1, k\ge 0.   
\]
        \item Given any $ 0<\epsilon<1/2 $, there exists $ \kappa _{ 1} >0$ such that 
        \[ 
 	\big| \Lambda _{ \xi} (u)-\Lambda _{ \xi}(v) \big|    \le \kappa _{ 1} \|u-v\|  \ \ \ \ \mbox{ whenever }  \|u\|\le M,\|v\|\le M  \mbox{ and } \| u-v \|\le \epsilon  ,
\]
where $ \|\cdot \|$ denotes the sup-norm on $ \mathbb{R}^{ K}$ and
\[ 
 	M=M^{ \prime } \Delta    \|\bar\beta \| \sum_{ k=- \infty}^{  \infty} \big| \phi_{ k} \big| ,
\]
\item   and  there exists $ L\ge 1$ such that $\sum_{ |k|>L}|\phi_{ k}|< \epsilon /(M^{ \prime }\Delta \|\bar\beta\|) $.
\end{enumerate}
We get (ii) since $ \Lambda _{ \xi}(\cdot)$ is convex and differentiable (cf. Lemma 2.2.5 \cite{dembo:zeitouni:1998}) and (iii) follows from the summability of the coefficients $ (\phi_{ k})$. 

 Define the function $ f_{ t,k}:(k,k+1)\to \mathbb{R}$ by
\[ 
 	f_{ t,k}(y):=    \Lambda _{ \xi} \left( \frac{t \lambda }{N((k+1)t)-N(kt) }  \Big(\sum_{ l=kt+1}^{ (k+1)t}\lfloor l^{ \alpha }\rfloor \phi_{ l- \lceil ty \rceil} \Big)  \beta\right).
\]
and note that 
\[ 
 	 \frac{ 1}{ t}  \sum_{ j=kt+1}^{ (k+1)t}  \Lambda _{ \xi} \left(\frac{ t\lambda }{N((k+1)t)-N(kt)  }  \beta \sum_{ l=kt+1}^{ (k+1)t}\lfloor l^{ \alpha }\rfloor \phi_{ l-j} \right)= \int_{ k}^{ k+1}f_{ t,k}(y)dy. 
\]
 Choose $ t$ large enough such that $ kt+1\le \lceil ty \rceil-L$, $ \lceil ty \rceil+L\le (k+1)t$  and 
\[ 
 	  \left| \frac{t\lfloor l^{ \alpha }\rfloor}{N((k+1)t)-N(kt) } - \frac{ (\alpha +1)y^{ \alpha }}{C((k+1)^{ \alpha +1}-k^{ \alpha +1}) }\right| \le \frac{ \epsilon }{ \Delta \|\beta \|    } \Big( \sum_{ k=- \infty}^{  \infty}\big| \phi_{ k} \big|  \Big)^{ -1}, 
\]
for all $ k\ge 0$, $ k+\epsilon <y<k+1-\epsilon $ and $ \lceil ty \rceil-L\le l\le \lceil ty \rceil+L$. It is easy to check that for $ y$ in this range and $ |\lambda|\le \Delta$
\[ 
 	\left\|   \frac{ t\lambda \beta }{N((k+1)t)-N(kt) }  \sum_{ l=kt+1}^{ (k+1)t}\lfloor l^{ \alpha }\rfloor \phi_{ l-\lceil ty \rceil}  -  \frac{ t\lambda \beta }{N((k+1)t)-N(kt) }\sum_{ l=\lceil ty \rceil-L}^{ \lceil ty \rceil+L}\lfloor l^{ \alpha }\rfloor \phi_{ l- \lceil ty \rceil} \right\|   \le\epsilon 
\]
and
\[ 
 	\left\|     \frac{ t\lambda \beta }{N((k+1)t)-N(kt) }\sum_{ l=\lceil ty \rceil-L}^{ \lceil ty \rceil+L}\lfloor l^{ \alpha }\rfloor \phi_{ l- \lceil ty \rceil} -  \frac{ (\alpha +1)\lambda y^{ \alpha } \bar \beta }{(k+1)^{ \alpha +1}-k^{ \alpha +1} }  \sum_{ l=\lceil ty \rceil-L}^{ \lceil ty \rceil+L} \phi_{ l- \lceil ty \rceil}\right\|   \le \epsilon 
\]
and
\[ 
 	\left\|    \frac{ (\alpha +1)\lambda y^{ \alpha } \bar \beta }{(k+1)^{ \alpha +1}-k^{ \alpha +1} }  \sum_{ l=\lceil ty \rceil-L}^{ \lceil ty \rceil+L}\phi_{ l- \lceil ty \rceil}   -  \frac{ (\alpha +1)\lambda y^{ \alpha } \bar \beta }{(k+1)^{ \alpha +1}-k^{ \alpha +1} }  \sum_{ l=- \infty}^{ \infty} \phi_{ l}\right\|   \le \epsilon .
\]
This implies for all $ k\ge 0$, $ k+\epsilon <y<k+1-\epsilon $ and $  | \lambda  |\le \Delta$
\[ 
 	\left| \Lambda _{ \xi} \left( \frac{(\alpha +1)\lambda \phi  y^{ \alpha }}{(k+1)^{ \alpha +1}-k^{ \alpha +1}}\bar\beta   \right) - f_{ t,k}(y)\right| \le    3 \kappa _{ 1}\epsilon 
\]
and hence we get 
\begin{equation}\label{eq:mid:part}
	     \lim_{ t \rightarrow \infty  } \sup_{ k\ge 0, | \lambda  |\le \Delta} \left|  \Lambda ^{ k}(\lambda ) - \frac{ 1}{ t}  \sum_{ j=kt+1}^{ (k+1)t}  \Lambda _{ \xi} \left(\frac{ t\lambda }{N((k+1)t)-N(kt)  }  \beta \sum_{ l=kt+1}^{ (k+1)t}\lfloor l^{ \alpha }\rfloor \phi_{ l-j} \right)\right|  \le 3\kappa _{ 1}\epsilon +4 M_{ 1}\epsilon,
 \end{equation}
 where 
\begin{equation}\label{eq:M1} 
 	M_{ 1}= \max \left\{ \Lambda _{ \xi} \left(  M^{ \prime }\Delta \|\beta\| \sum_{ k=- \infty}^{ \infty}\big| \phi_{ k} \big|    \right), \Lambda _{ \xi} \left(  -M^{ \prime }\Delta \|\beta\| \sum_{ k=- \infty}^{ \infty}\big| \phi_{ k} \big|    \right)  \right\}   .
\end{equation}
Obviously, since $ \epsilon $ is arbitrary we get that the limit in \eqref{eq:mid:part} is 0.
 
  The other parts in \eqref{eq:triangle} are handled much easily. Note that for any $ k\ge 0$
\[ 
 	  \left| \sum_{ j=- \infty}^{  kt-L}  \Lambda _{ \xi} \left(\frac{ t\lambda }{N((k+1)t)-N(kt)  }  \beta \sum_{ l=kt+1}^{ (k+1)t}\lfloor l^{ \alpha }\rfloor \phi_{ l-j} \right)\right|\le \kappa _{ 1} M^{ \prime } \Delta \| \bar\beta\|  \sum_{ j=- \infty}^{kt  -L}  \sum_{ l=kt+1}^{ (k+1)t} \big| \phi_{ l-j} \big|  \le t \kappa _{ 1} \epsilon
\]
and hence
\begin{equation}\label{eq:k:le:L}
	\lim_{ L \rightarrow \infty  } \lim_{ t \rightarrow \infty  } \sup_{ k\ge 0, | \lambda  |\le \Delta}\left|\frac{ 1}{t } \sum_{ j=- \infty}^{  kt-L}  \Lambda _{ \xi} \left(\frac{ t\lambda }{N((k+1)t)-N(kt)  }  \beta \sum_{ l=kt+1}^{ (k+1)t}\lfloor l^{ \alpha }\rfloor \phi_{ l-j} \right)\right| =0.
\end{equation}
Using a similar argument we also get
\begin{equation}\label{eq:k:ge:nplusL}
	\lim_{ L \rightarrow \infty  }\lim_{ t \rightarrow \infty  } \sup_{ k\ge 0, | \lambda  |\le \Delta}\left|\frac{ 1}{t } \sum_{ j=(k+1)t+L}^{ \infty}  \Lambda _{ \xi} \left(\frac{ t\lambda }{N((k+1)t)-N(kt)  }  \beta \sum_{ l=kt+1}^{ (k+1)t}\lfloor l^{ \alpha }\rfloor \phi_{ l-j} \right)\right| =0.
\end{equation}
Furthermore, it is also easy to check that for every $ L\ge 1$
\begin{align}\label{eq:small:parts}
	& \lim_{ t \rightarrow \infty  }  \sup_{ k\ge 0, | \lambda  |\le \Delta} \left| \frac{ 1}{ t}\sum_{ \substack{kt-L<j\le kt \mbox{ or}\\(k+1)t<j\le (k+1)t+L}}  \Lambda _{ \xi} \left(\frac{ t\lambda }{N((k+1)t)-N(kt)  }  \beta \sum_{ l=kt+1}^{ (k+1)t}\lfloor l^{ \alpha }\rfloor \phi_{ l-j} \right)\right|\nonumber\\
	& \le \lim_{ t \rightarrow \infty  } \frac{ 1}{ t} 2(L+1) M_{ 1}=0.
\end{align}

For the final part of the proof of \eqref{eq:unif:conv}, we note  the following facts about $ \Lambda _{ \varepsilon }(\cdot) $: $ \Lambda _{ \varepsilon }(0)=0$, $\Lambda _{ \varepsilon }^{ \prime }(0)=0 $ because $ E(\varepsilon_{ i,j}(t))=0$, $ \Lambda_{ \varepsilon }(\cdot) $ is nonnegative and twice continuously differentiable in a neighborhood of $ 0$. The last fact can be easily derived following Lemma 2.2.5 in \cite{dembo:zeitouni:1998}. This implies that there exist positive constants $ \kappa $ and $ \eta$ such that 
\[ 
 	\big|\Lambda _{ \varepsilon }(u )   \big|\le \kappa u ^{ 2} \ \ \ \ \mbox{ for all } | u  |\le \eta.  
\]
Choose $ t$ large enough such that $  t\Delta /N(t) < \eta  $. This also means that $ |t \lambda /(N((k+1)t)-N(kt)) |< \eta  $  for all $ k\ge 0$ and $  | \lambda  |\le \Delta$. Hence, we have
\begin{align*}
	&   \left| \log E\left[ \exp \left\{ \frac{ t\lambda }{N((k+1)t)-N(kt) }\sum_{ l=1}^{ t}\sum_{ i=1}^{ K}\sum_{ j=1}^{ n_{ i}(l)}\varepsilon_{ i,j}(l)   \right\}  \right] \right|  \\
	& = \sum_{ l=kt+1}^{(k+1)t }\sum_{ i=1}^{ K}\sum_{ j=1}^{ n_{ i}(l)} \Lambda _{ \varepsilon }\left( \frac{t \lambda }{ N((k+1)t)-N(kt) } \right) \\
	 &\le   \big( N((k+1)t)-N(kt)  \big) \kappa \frac{ t^{ 2} \lambda ^{ 2}}{ \big( N((k+1)t)-N(kt)  \big)^{ 2} },
 \end{align*}
  This immediately gives us
\begin{align}\label{eq:err:part}
	&\lim_{ t \rightarrow \infty }\sup_{ k\ge 0,|\lambda |\le \Delta} \left| \frac{ 1}{ t}\log E\left[ \exp \left\{ \frac{ \lambda }{N((k+1)t)-N(kt) }\sum_{ l=1}^{ t}\sum_{ i=1}^{ K}\sum_{ j=1}^{ n_{ i}(l)}\varepsilon_{ i,j}(l)   \right\}  \right] \right| \nonumber \\	
	& \le \lim_{ t \rightarrow \infty } \kappa\sup_{ k\ge 0, | \lambda  |\le \Delta} \frac{ t \lambda }{ \big( N((k+1)t)-N(kt)  \big) }=0.
\end{align}
and that completes the proof of \eqref{eq:unif:conv}.

It is simpler to check the other conditions of Theorem \ref{thm:gen:uldp}. Note that we can find $ M$ such that 
\begin{equation}\label{eq:unif:bdd:coeff} 
 	\frac{ y^{ \alpha }}{(k+1)^{ \alpha +1}-k^{ \alpha +1} }  \le M \ \ \ \ \mbox{ for all }k\ge 0, k\le y\le k+1.
\end{equation}
This implies that for any $ \Delta>0$
\[ 
 	\sup_{ k\ge 0,  | \lambda  |\le \Delta}\left| \Lambda ^{ k}(\lambda ) \right|   <\infty,
\]
and this combined with \eqref{eq:unif:conv} shows that the condition \eqref{eq:cond2} holds.

Next we check that $ \Lambda ^{ k}(\cdot)$ is differentiable. Since $ \Lambda_{ \xi}(\cdot)$ is finite everywhere, by Lemma 2.2.5 in \cite{dembo:zeitouni:1998} we get that  $ \Lambda _{ \xi}(\cdot)$ is differentiable and 
\[ 
 	\Lambda _{ \xi}^{ \prime }(\eta )=\frac{ E\big[ \xi(0)e^{ \eta \cdot \xi(0)} \big] }{ E\big[ e^{ \eta \cdot \xi(0)} \big]  }   
\] 
 For any $ \delta $ satisfying $ 0<\|\delta\|<1 $ 
\[ 
 	ze^{ (\eta +\delta )\cdot z}-ze^{ \eta\cdot z}\to 0 \ \  \mbox{ and  }   \ \ \|ze^{ (\eta +\delta )\cdot z}-ze^{ \eta\cdot z}\|\le h(z):=\|z\|e^{ \eta \cdot z}(e^{  \|z\|}+1).
\]
Since $ E[h(\xi(0))]<\infty$ using the dominated convergence theorem we get that $ E[\xi(0)e^{ \lambda \xi(0)}]$ is  continuous. This implies that $ \Lambda ^{ \prime }_{ \xi}(\cdot )$ is continuous. Now we can use the Leibniz integral rule (cf. Theorem 7.40 in \cite{apostol1974mathematical}) to get that $ \Lambda ^{ k}(\cdot)$ is differentiable and 
\[ 
 	(\Lambda ^{ k})^{ \prime }(\lambda ):=\int_{ k}^{ k+1} \frac{(\alpha +1) \phi  y^{ \alpha }}{(k+1)^{ \alpha +1}-k^{ \alpha +1}}\bar\beta\cdot\Lambda _{ \xi}^{ \prime } \left( \frac{(\alpha +1)\lambda \phi  y^{ \alpha }}{(k+1)^{ \alpha +1}-k^{ \alpha +1}}\bar\beta   \right)dy 
\]
It is easy to see that $ \|\Lambda _{ \xi}^{ \prime }(\eta)\|\to \infty$ whenever $ \|\eta\|\to \infty$. This combined with \eqref{eq:unif:bdd:coeff} shows that \eqref{eq:cond3} holds. Finally, \eqref{eq:cond:uniform:cont} follows from the fact that $ \Lambda _{ \xi}^{ \prime }(\cdot)$ is continuous on compact sets and \eqref{eq:unif:bdd:coeff}. This completes the proof of the theorem.
\end{proof}

\end{section}

\begin{section}{Proof of Theorem \ref{thm:lss} and Required Lemmas}\label{sec:lss}

\begin{proof}[Proof of Theorem \ref{thm:lss}.]
We will first prove the lower inequality in \eqref{eq:Trresult}.  The inequality is obvious when $ I_{ *}=0$. Also, if $ A$ is nonempty then $ I_{ *}< \infty$ from Assumption \ref{assmp:steep}. So it suffices to consider  $ 0<I_{ *}< \infty$. We will use the simple  inclusion bound: for all $ m\ge 1$ and $ r\ge 1$
\[ 
 	   \{   T_{ r}(A)\le m  \}    \subset \bigcup_{ l=r}^{ \infty}   \bigcup_{ j=0}^{m -1} \left\{ \bar X(j,j+l)\in A \right\}.
\]
Thus we get
\[
	P\big[ T_{ r}(A)\le m \big] \le  \sum_{ l=r}^{ \infty} \sum_{ j=0}^{ m-1} P\big[ \bar X(j,j+l)\in A . \big] 
\]
Lemma~\ref{lem:property:Lambda} below shows that the $\Lambda^{k*}(x)$ are an increasing function of $k$ for fixed $x$. Lemma~\ref{lem:ratefn:largek}, which builds on this, gives the existence of a $ K_{ 0}$ such that 
\[ 
 	\inf_{ x\in \bar A} \Lambda ^{ k*}(x)\ge \inf_{ x\in \bar A} \Lambda^{ *} (x) -\epsilon /3=I_{ *}-\epsilon /3  \ \ \ \ \mbox{ for all }k\ge K_{ 0} .
\]
We can also find, from Lemma \ref{lem:max:ratefn}, a constant $  I>0$ such that  $  I\le \inf_{ x\in \bar A}\Lambda ^{ k*}(x)$ for all $ k\ge 0$.
Now for any $0< \epsilon <I$, by Theorem \ref{thm:uldp} we can get  $ T\ge 1$ such that for all $ l\ge T$ and all $ k\ge 0$
\[ 
 	   P\Big[ \bar X \big( kl,(k+1)l \big)\in A  \Big] \le \exp  \left\{ -l \left(\inf_{ x\in \bar A}\Lambda ^{ k*}(x) -\epsilon/3 \right)\right\}
\]
This gives us for $ r\ge T$
\begin{align*}
	P\big[ T_{ r}(A)\le m \big]&\le     \sum_{ l=r}^{ \infty} \sum_{ j=0}^{ m-1} P\Big[ \bar X(j,j+l)\in A  \Big] \\
	&  \le \sum _{ l=r}^{ \infty} \sum_{ j=0}^{ K_{ 0}l} P\Big[ \bar X(j,j+l)\in A  \Big]+ \sum _{ l=r}^{ \infty} \sum_{ j=K_{ 0} l}^{ m-1} P\Big[ \bar X(j,j+l)\in A  \Big]\\
	& \le   \sum _{ l=r}^{ \infty}K_{ 0}l e^{ -l( I-\epsilon /3)} +m  \sum _{ l=r}^{ \infty}e^{ -l(I_{ *}-2\epsilon /3)}.
 \end{align*}
Now set $ m=\lfloor e^{ r(I_{ *}-\epsilon )} \rfloor$ and note that
\begin{align*}
	&\sum_{ r=1}^{  \infty} P\big[ T_{ r}(A)\le\lfloor e^{ r(I_{ *}-\epsilon )} \rfloor\big]     \\
	& \le T+ \sum_{ r=T}^{  \infty} \sum _{ l=r}^{ \infty}K_{ 0}l e^{ -l( I-\epsilon /3)} +\sum_{ r=T}^{  \infty}  e^{ r(I_{ *}-\epsilon )} \sum _{ l=r}^{ \infty}e^{ -l(I_{ *}-2\epsilon /3)}< \infty.
 \end{align*}
 Hence, using the Borel-Cantelli lemma we get  
 \[ 
 	   \liminf_{ r \rightarrow  \infty} \frac{ \log T_{ r}(A)}{r }\ge I_{ *}-\epsilon  \ \ \ \ a.s.
\]
The lower inequality in \eqref{eq:Trresult} is thus proved by letting $ \epsilon \to 0$.

Also observe that, using the relation $ \{   T_{ r}(A)\le m  \}=\{   R_{ m}(A)\ge r  \}$ we get 
\[ 
 	\limsup_{ t \rightarrow \infty  } \frac{ R_{ t}}{ \log t}\le \frac{ 1}{I_{ *} } \ \ \ \  a.s.   
\]

In order to prove the upper bound in \eqref{eq:Trresult} it suffices to consider 
the case $I^*<\infty$. In that case the set $A$ has nonempty interior.
Define two new random variables by  
\[
	Y_{ k,t}^{ \prime }:= \beta \frac{\sum_{ j=kt+1}^{ (k+1)t} \sum_{ l=kt+1}^{ (k+1)t} \lfloor l^{ \alpha }\rfloor \phi_{ l-j} \xi(j)} { N((k+1)t)-N(kt)} \ \ \mbox{ and } \ \ Y_{ k,t}^{ \prime \prime }:=\bar X(kt,(k+1)t)-Y_{ k,t}^{ \prime },
\]
where, as before, $\beta =\sum_{ i=1}^{ K}c_{ i}\beta _{ i}$. 
 For a set $A$ and $\eta>0$, define
\[
A(\eta) :=\big\{x:d(x,A^c)>\eta\big\},
\]
and $d(x,A^c)$ is the distance from the point $x$ to the complement  $A^c$. Now observe that for any positive integers $ r$ and $ q$ with $ q>r$
\begin{align*}
	& P\big[T_r(A)>q\big]\\
	 & \le  P\left[  \bar X(kr,(k+1)r)\notin A, k=0,\ldots, \left\lfloor
	 q/r\right\rfloor  \right] \\ 
& \le  P\left[  Y_{ k,r}^{ \prime }\notin A(\eta),k=0,\ldots, \left\lfloor
	 q/r\right\rfloor \right] 
     + \sum_{l=1}^{\lfloor q/r \rfloor} P\left[
	 |Y_{ k,r}^{ \prime \prime }|>\eta\right]\\ 
	 \intertext{and since $ Y_{ k,r}^{ \prime },k=0,1,\ldots \lfloor q/r\rfloor$ are independent}
&= \prod_{ k=0}^{ \lfloor q/r\rfloor} \Big( 1-P\left[  Y_{ k,r}^{ \prime }\in A(\eta) \right] \Big) +  \sum_{l=1}^{\lfloor q/r \rfloor} P\left[
	 |Y_{ k,r}^{ \prime \prime }|>\eta\right]\\ 
& \le  \exp\Big(-\sum_{ k=0}^{ \lfloor q/r\rfloor}P\left[  Y_{ k,r}^{ \prime }\in A(\eta) \right] \Big)+
	 \sum_{l=1}^{\lfloor q/r \rfloor} P\left[
	 |Y_{ k,r}^{ \prime \prime }|>\eta\right]. 
\end{align*}

From the arguments following \eqref{eq:mid:part} it is easy to check that the law of $Y_{ k,t}^{ \prime } $ satisfy large deviation principle uniformly over $ k\ge 0$ with rate function $ \Lambda ^{ k*}(\cdot)$. We can therefore, get $ T\ge 1$ such that
\[ 
 	   \frac{ 1}{t } \log P\big[ Y_{ k,t}^{ \prime } \in A(\eta) \big] \ge - \inf_{ x\in A(\eta)}\Lambda ^{ k*}(x) -\epsilon /4 \ \ \ \ \mbox{ for all } t\ge T, k\ge 0.
\]
Lemma \ref{lem:property:Lambda}\emph{(ii)} then implies 
\[ 
 	   \frac{ 1}{t } \log P\big[ Y_{ k,t}^{ \prime } \in A(\eta) \big] \ge - \inf_{ x\in A(\eta)}\Lambda ^{ *}(x) -\epsilon /4 \ \ \ \ \mbox{ for all } t\ge T, k\ge 0.
\]
Hence for $ \eta>0$ small enough 
\[ 
 	   \frac{ 1}{t } \log P\big[ Y_{ k,t}^{ \prime } \in A(\eta) \big] \ge -I^{ *} -\epsilon /2 \ \ \ \ \mbox{ for all } t\ge T, k\ge 0.
\]
Therefore, by setting $ q_{ r}=\lceil e^{ r(I^{ *}+\epsilon )}\rceil$ and using the above inequality we get that 
\begin{align}\label{eq:lastbut1part}
	& \sum_{ r=1}^{ \infty}     \exp\Big(-\sum_{ k=0}^{ \lfloor q_{ r}/r\rfloor}P\left[  Y_{ k,r}^{ \prime }\in A(\eta) \right] \Big) \nonumber\\
	&  \le T+ \sum_{ r=T}^{ \infty} \exp \Big( - \frac{e^{ r(I^{ *}+\epsilon )} }{ r} e^{ -r(I^{ *}+\epsilon /2)} \Big)\nonumber\\
	& \le T +\sum_{ r=T}^{ \infty} \exp \Big( -\frac{e^{ r\epsilon /2}}{r}  \Big)< \infty.
 \end{align}
Furthermore, note that for $ \epsilon >0$ and $ \eta>0$ such that the above holds
\begin{align*}
	& \limsup_{ t \rightarrow \infty  }  \sup_{ k\ge 0}  \frac{ 1}{ t} \log P\big[  | Y_{ k,t}^{ \prime \prime } | > \eta\big] \\
	&  \le -\lambda \eta + \limsup_{ t \rightarrow \infty  }  \sup_{ k\ge 0} \frac{ 1}{t }\log E\big[ \lambda t |Y_{ k,t}^{ \prime \prime } \big] =- \lambda \eta.
 \end{align*}
 The last equality follows from the steps used in the proof of Theorem \ref{thm:uldp}. Now by choosing $ \lambda >(I^{ *}+\epsilon )/\eta$ we get
\begin{equation}\label{eq:lastpart}
    \sum_{ r=1}^{ \infty} \sum_{l=1}^{\lfloor q_{ r}/r \rfloor} P\left[ |Y_{ k,r}^{ \prime \prime }|>\eta\right]
  \le \sum_{ r=1}^{ \infty} \Big\lfloor \frac{q_{ r} }{r } \Big\rfloor \sup_{ k\ge 0}P\left[ |Y_{ k,r}^{ \prime \prime }|>\eta\right]< \infty.
\end{equation}
Combining \eqref{eq:lastbut1part} and \eqref{eq:lastpart} we get 
\[ 
 	   \sum_{ r=1}^{ \infty}P\big[T_r(A)>q\big]<\infty.
\]
Finally by applying the first Borel-Cantelli lemma and then letting $ \epsilon \to 0$ we complete the proof of the upper bound of 
\eqref{eq:Trresult}. The lower bound in \eqref{eq:Rtresult} is again proved using the same identity $ \{   T_{ r}(A)\le m  \}=\{   R_{ m}(A)\ge r  \}$. Hence the proof is complete.
\end{proof}

\begin{lemma}\label{lem:property:Lambda}
\begin{enumerate}[(i)] 
\item For any $ \lambda \in \mathbb{R}$, $ \Lambda ^{ k}(\lambda )$ is a decreasing function of $ k$.
        \item           For any $ x\in \mathbb{R}$, $ \Lambda ^{ k*}(x)$ is an increasing function of $ k$. 
\end{enumerate}
\end{lemma}

\begin{proof} 
Suppose $ F_{ k}$ is the distribution function of the random variables 
\[ 
 	U_{ k}:=\frac{ (\alpha +1)(k+U)^{ \alpha }}{ (k+1)^{ \alpha +1}-k^{ \alpha +1} } \ \ \ \ \mbox{ where }U\sim \mbox{Uniform}(0,1),k\ge 0.   
\]
Observe that $ E(U_{ k})=1$ for all $ k\ge 0$. Also, for any non-negative random variable $ X$ with mean 1 and distribution $ F_{ X}$, define the Lorenz function
\[ 
 	L_{ X}(p)  :=   \int_{ 0}^{p} F_{ X}^{ -1}(u) du, \ \ \ \ \mbox{ for all }0\le p\le 1. 
\]
Note that the Lorenz function of $ U_{ k}$ is given by
\[ 
 	L_{ U_{ k}}(p)= \frac{ (k+p)^{ \alpha +1}-k^{ \alpha +1}}{ (k+1)^{ \alpha +1} -k^{ \alpha +1} },   \ \ \ \ \mbox{ for all } 0\le p\le 1.
\]
and 
\[ \frac{\partial}{\partial k} L_{ U_{ k}}(p)= \frac{(\alpha +1) \big[(k+p)^{ \alpha }\left( (k+1)^{ \alpha }(1-p)+k^{ \alpha }p \right)-k^{ \alpha }(k+1)^{ \alpha }  \big]  }{ ((k+1)^{ \alpha +1}-k^{ \alpha +1})^{ 2} }>0  \] 
for all  $ k\ge 0$  and  $ 0\le p\le 1$. This implies that 
\[ 
 	L_{ U_{ k^{ \prime }}}(p)\ge L_{ U_{ k^{ \prime \prime }}}(p) \ \ \ \ \mbox{ for all }0\le p\le 1, k^{ \prime }\ge k^{ \prime \prime }
\]
which means that $ U_{ k}$ is decreasing in Lorenz order as $ k$ increases. Hence by \citep[Theorem 3.2, p.37]{Arnold.B:8000} and using the fact that $ \Lambda_{ \xi} (\cdot )$ is convex and continuous we get that 
\[ \Lambda ^{ k}(\lambda ) = E\big[ \Lambda _{ \xi}\big(\lambda \phi U_{ k}\bar \beta\big)  \big]  \]
 is decreasing in $ k$. 
 
 Part \emph{(ii)} of the lemma follows easily from part \emph{(i)} using the definition of Fenchel-Legendre transform.
\end{proof}

\begin{lemma}\label{lem:ratefn:largek} 
 For any measurable set $ A\subset \mathbb{R}$ and $ \epsilon >0$ there exists $ K_{ 0}$ such that 
 \[ 
 	 \inf_{ x\in A}\Lambda ^{ k*}(x) \ge \inf_{ x\in A}\Lambda ^{ *}(x)-\epsilon  \ \ \ \ \mbox{ for all }k\ge K_{ 0}
\]
where $ \Lambda ^{ k*}(\cdot)$ and $ \Lambda ^{ *}(\cdot)$ are as described in Theorem \ref{thm:uldp} and Theorem \ref{thm:lss}, respectively.
\end{lemma}

\begin{proof} 
Fix any $ \epsilon >0$. From the arguments leading to \eqref{eq:uniform:good:ratefn} we can find $ M_{ 1}>0$ such that 
\[ 
 	\inf_{ x\in A} \Lambda ^{ *}(x) =\inf_{ x\in A\cap [-M_{ 1},M_{ 1}]} \Lambda ^{ *}(x)   .
\]
Lemma \ref{lem:property:Lambda}\textit{(ii)} then gives us
\[ 
	\inf_{ x\in A} \Lambda ^{k *}(x) =\inf_{ x\in A\cap [-M_{ 1},M_{ 1}]} \Lambda ^{k *}(x)   \ \ \ \ \mbox{ for all }k\ge 0.
\]
Using Assumption \ref{assmp:steep} we get $ M_{ 2}>0$ such that $ |\lambda |>M_{ 2}$ implies $ |(\Lambda ^{ 0})^{ \prime }(\lambda )|>2M_{ 1}$. Since $ \Lambda ^{ k}(\cdot)$ converges locally uniformly to $ \Lambda (\cdot)$ we know that there exists $ K_{ 0}$ such that 
\begin{equation}\label{eq:1}
	\sup_{ \lambda \in [-M_{ 2},M_{ 2}]} \left| \Lambda ^{ k}(\lambda )-\Lambda (\lambda ) \right| <\epsilon /4 \ \ \ \ \mbox{ for all }k\ge K_{ 0}.
\end{equation}
Now, for any $ x\in [-M_{ 1},M_{ 1}]$ we can get $ \lambda _{ x}\in [-M_{ 2},M_{ 2}]$  such that $ \lambda _{ x}x-\Lambda (\lambda _{ x}) >\Lambda ^{ *}(x)-\epsilon /4$  and therefore for all $ k\ge K_{ 0}$
\[ 
 	\Lambda ^{ k*}(x)\ge \lambda _{ x}x-\Lambda ^{ k}(\lambda _{ x})\ge    \lambda _{ x}x-\Lambda (\lambda _{ x})-\epsilon /4\ge \Lambda ^{ *}(x)-\epsilon /2.
\]
This implies for all $ k\ge K_{ 0}$
\[ 
 	 \inf_{ x\in A\cap [-M_{ 1},M_{ 1}]}\Lambda ^{ k*}(x) \ge \inf_{ x\in A\cap [-M_{ 1},M_{ 1}]}\Lambda ^{ *}(x)-\epsilon  
\]
and that completes the proof.
\end{proof}

\begin{lemma}\label{lem:max:ratefn} 
For any measurable set $ A\subset \mathbb{R}$
\begin{equation}\label{eq:lem:max:ratefneq} 
 	 \inf_{ x\in A}\Lambda ^{ *}(x)>0 \ \ \ \ \mbox{implies} \ \ \ \ \inf _{ k\ge 0}\inf_{ x\in A}\Lambda ^{ k*}(x)>0.   
\end{equation}
\end{lemma}

\begin{proof} 
Using Lemma \ref{lem:property:Lambda}\emph{(ii)} it suffices to show that \eqref{eq:lem:max:ratefneq} implies $ \inf_{ x\in A}\Lambda ^{ 0*}(x)>0$. Fix any $ x\neq 0$. Since $ \Lambda ^{0 }(\lambda )$ is strictly convex and finite everywhere and  $(\Lambda^0)^{\prime}(0) = 0$,   we get that if  $ (\Lambda ^{ 0})^{ \prime }(\lambda ^{ 0}_{ x})=x$ then $ \lambda ^{ 0}_{ x}\ne 0$.  Then $ \Lambda ^{ 0*}(x)=\lambda ^{ 0}_{ x}x-\Lambda ^{ 0}(\lambda ^{ 0}_{ x})\ne 0$. If for some measurable $ A\subset \mathbb{R}$ then 
\[ 
 	\inf_{ x\in A}\Lambda ^{ 0*}(x)=0 \ \ \mbox{ implies } \ \ 0\in \bar A.   
\]
That would imply $ \inf_{ x\in A}\Lambda ^{ *}(x)=0$. This proves the lemma.
\end{proof}
\end{section}

\bibliographystyle{elsart-harv}
\bibliography{/Users/Souvik/Documents/Chronicles/Miscellanea/Latex/bibfile}
\end{document}